\documentclass[prb,twocolumn,amsmath,amssymb,10pt,aps,longbibliography,superscriptaddress,citeautoscript,bibnotes,floatfix]{revtex4-2}

\usepackage{graphicx} 
\usepackage{dcolumn}
\usepackage{bm}
\usepackage{graphicx,color,xcolor}
\usepackage{physics}
\usepackage{comment}
\usepackage[normalem]{ulem}
\graphicspath{ {./Figures/} }
\usepackage{textcomp}
\usepackage[utf8]{inputenc}
\usepackage[T1]{fontenc}
\usepackage{amsmath}
\usepackage{amssymb}
\usepackage{amsthm}
\usepackage{siunitx}
\usepackage{physics}
\usepackage{graphicx}
\usepackage{subfigure}
\usepackage{tabularx}
\usepackage{booktabs}
\usepackage{multirow}
\usepackage[colorlinks=true,linkcolor=blue,citecolor=blue]{hyperref} 
\usepackage{soul}
\usepackage{xcolor}
\usepackage{csquotes}

\usepackage{algorithm}
\usepackage{algpseudocode}

\newtheorem{theorem}{Theorem}

\newtheorem{lemma}{Lemma}

\begin{document}

\title{Fast pseudothermalization}

\author{Wonjun Lee}
\email{wonjun1998@postech.ac.kr}
\affiliation{Department of Physics, Pohang University of Science and Technology, Pohang, 37673, Republic of Korea}
\affiliation{Center for Artificial Low Dimensional Electronic Systems, Institute for Basic Science, Pohang, 37673, Republic of Korea}
\author{Hyukjoon Kwon}
\email{hjkwon@kias.re.kr}
\affiliation{School of Computational Sciences, Korea Institute for Advanced Study, Seoul 02455, South Korea}
\author{Gil Young Cho}
\email{gilyoungcho@gmail.com}
\affiliation{Department of Physics, Korea Advanced Institute of Science and Technology, Daejeon 34141, Korea}
\affiliation{Center for Artificial Low Dimensional Electronic Systems, Institute for Basic Science, Pohang, 37673, Republic of Korea}
\affiliation{Asia-Pacific Center for Theoretical Physics, Pohang, 37673, Republic of Korea}

\date{\today} 

\begin{abstract}
Quantum resources like entanglement and magic are essential for characterizing the complexity of quantum states. However, when the number of copies of quantum states and the computational time are limited by numbers polynomial in the system size $n$, accurate estimation of the amount of these resources becomes difficult. This makes it impossible to distinguish between ensembles of states with relatively small resources and one that has nearly maximal resources. Such ensembles with small resources are referred to as \enquote{pseudo-quantum} ensembles. Recent studies have introduced an ensemble known as the random subset phase state ensemble, which is pseudo-entangled, pseudo-magical, and pseudorandom. While the current state-of-the-art implementation of this ensemble is conjectured to be realized by a circuit with $O(nt)$ depth, it is still too deep for near-term quantum devices to execute for small $t$. In addition, the strict linear dependence on $t$ has only been established as a lower bound on the circuit depth. In this work, we present significantly improved implementations that only require $\omega(\log n)\cdot O(t[\log t]^2)$ depth circuits, which almost saturates the theoretical lower bound. 
This is also the fastest known for generating pseudorandom states to the best of our knowledge. We believe that our findings will facilitate the implementation of pseudo-ensembles on near-term devices, allowing executions of tasks that would otherwise require ensembles with maximal quantum resources, by generating pseudo-ensembles at a super-polynomially fewer number of entangling and non-Clifford gates.
\end{abstract}

\maketitle

\section{Introduction}
Randomness is a fundamental feature of quantum mechanics. At its core, quantum mechanics involves randomness when transforming a quantum state into classical data through measurements. Moreover, random processes in quantum mechanics, such as random circuit evolution and chaotic dynamics, have been widely considered in connection with various problems in recent years, from understanding the nature of black holes~\cite{yoshida2017,Yoshida2019,Piroli_2020} to demonstrating quantum supremacy~\cite{Boixo_2018,arute2019} or characterizing quantum states in advanced quantum devices~\cite{Knill_2008,Huang_2020,Helsen2022}. The properties of these random quantum objects are described statistically. Thus, to uncover their characteristics, one must prepare multiple copies and perform measurements to learn incrementally from each attempt. Ideally, one would need an infinite number of copies with arbitrary measurement setups to fully know the properties of these random objects. However, in practice, this is not feasible, as the number of realizable copies as well as the executable time of quantum algorithms are limited by a number polynomial in the system size. Recent research has explored how well the properties of random quantum objects can be understood under these practical constraints~\cite{Ji2018,aaronson2023quantum,gu2023little,haug2023pseudorandom,lu2024,gu2024chaos,lee2024pseudo}. For instance, when random processes are coupled with the randomness of measurements, it becomes increasingly difficult to identify the source of randomness. Consequently, the new notion called the pseudorandom state ensemble has been developed which is indistinguishable from the maximally random state ensemble by analyzing only a polynomial number of measurement outcomes~\cite{Ji2018}. 

The random subset phase state ensemble is a notable example of a pseudorandom state ensemble~\cite{aaronson2023quantum}. Beyond its pseudorandom nature, this ensemble exhibits two remarkable properties: pseudo-entangled and pseudo-magical~\cite{aaronson2023quantum,gu2023little}. Basically, it cannot be distinguished from an ensemble possessing maximal entanglement or magic, respectively, through a polynomial number of measurements. These pseudonesses arise from challenges in precise quantification of them within a quantum state via measurements~\cite{Soleimanifar2022,Oliviero_2022,Leone2022,garcia2024}. The construction of this phenomenal ensemble in quantum devices has been developed recently. First, it can be implemented by a quantum-secure pseudorandom permutation, which can be executed in polynomial time~\cite{zhandry2016,aaronson2023quantum}. Second, it can also be generated by random controlled-controlled NOT (CCX) gate circuits~\cite{feng2024}. The latter method is conjectured to require a runtime $O(n^2 t)$ with the system size $n$ and the number of copies $t$.

In this work, we introduce new algorithms for generating a random subset phase state ensemble by applying the Multi-Contolled NOT (MCX) gate randomly under structured sequences of control and target positions. This gate with $m$ control bits becomes the CCX gate when $m=2$ and can be implemented by $O(m)$ CCX gates for arbitrary $m$~\cite{Barenco1995}. We find that our algorithms thermalizes $t$ distinct copies of the initial subset $\{0,1\}^k\times0^{n-k}$ for some $k=\omega(\log n)$ in a circuit of $\omega(\log n)O(t[\log t]^2)$ depth or $O(nt\log (n)\log(t))$ gates, which significantly improves previously proposed methods for generating pseudorandom states or approximate state $t$-design~\cite{Brand_o_2016,Nakata2017,Harrow_2023,Cotler2023,feng2024,haah2024,schuster2024}. Furthermore, we provide an algorithm that implements the random phases of a random subset phase state in a shorter depth. These can be directly used to implement random isometries and thus pseudo-chaotic dynamics in ~\cite{lee2024pseudo}.

\section{Random subset phase state}
A random subset phase state with the system size $n$ and the subset size $k=\omega(\log n)$ is built using two random mappings: random permutation $p:\{0,1\}^k\times\{0,1\}^{n-k}\rightarrow\{0,1\}^n$ and random function $f:\{0,1\}^k\times\{0,1\}^{n-k}\rightarrow\{0,1\}$. Formally, the state is given by
\begin{equation}\label{eq:random-subset-phase-state}
    \ket{\psi_{p,f,a}} = \frac{1}{2^{k/2}} \sum_{b\in\{0,1\}^k}(-1)^{f(b,a)}\ket{p(b,a)} 
\end{equation}
for some $a\in\{0,1\}^{n-k}$. Let $\mathcal{P}$ and $\mathcal{F}$ be the sets of all permutations and functions, respectively. A random subset phase state ensemble is then defined as $\{\ket{\psi_{p,f,a}}|p\sim\mathcal{P},f\sim\mathcal{F}\}$ with the uniform distributions over $\mathcal{P}$ and $\mathcal{F}$. This ensemble is pseudo-random, pseudo-entangled, and pseudo-magical~\cite{Ji2018,aaronson2023quantum,gu2023little}. $\mathcal{P}$ with a fixed $a$ has an one-to-one correspondence $s:\mathcal{P}\rightarrow\mathcal{S}_{2^k}$ to the set $\mathcal{S}_{2^k}$ of all $2^k$-dimensional subsets of $\{0,1\}^n$ such that $s(p) = \{\ket{p(b,a)}|b\in\{0,1\}^k\}$. As they are isomorphic, we will consider random subsets instead of random permutations and drop the dependence on $a$ hereafter. 

Now, let us consider an ensemble $\mathcal{E}$ of random subset phase states with a non-uniform distribution $\tilde{\pi}$ over $\mathcal{S}$ and the uniform distribution over $\mathcal{F}$. Then, the trace distance between the ensemble average of $\mathcal{E}$ and that of the Haar random ensemble is upper bounded by how close $\Phi_{t\leftarrow 2^k}[\tilde{\pi}]$ is to the uniform distribution for all integers $1\leq t \leq 2^k$ due to \textbf{Proposition 1} of~\cite{feng2024}. Here, $\Phi_{t\leftarrow 2^k}[\tilde{\pi}]$ is a stochastic map defined by
\begin{equation}
    \Phi_{t\leftarrow 2^k}[\tilde{\pi}](S') = \binom{2^k}{t}^{-1}\sum_{\{S\in \mathcal{S}_{2^k}|S'\subset S\}} \tilde{\pi}(S)
\end{equation}
with $S'\in\mathcal{S}_t$. More precisely, if $\Phi_{t\leftarrow 2^k}[\tilde{\pi}]$ becomes the uniform distribution with a negligible total variation distance, then $\mathcal{E}$ approximates the $t$-th moments of the Haar random ensemble up to a negligible trace distance. The uniformization of $\Phi_{t\leftarrow 2^k}$ is equivalent to the infinite temperature thermalization, \textit{i.e.}, the full randomization, of all bits of $t$ distinct copies of $\{0,1\}^n$. In this work, we study algorithms that thermalize initial random $t$ distinct copies of $\{0,1\}^k\times 0^{n-k}$ by applying MCX gates randomly. In addition, we also study an algorithm that thermalizes the sign factors $\{(-1)^{f(b,a)}\}_{b\in\{0,1\}^k}$ associated with each copy using Multi-Controlled $Z$ (MCZ) gate. 

Below, we introduce two different algorithms which uniformize $\tilde{\pi}$. \textbf{Algorithm}~\ref{alg:gate_num_opt} is optimized to minimize the number of required gates, while \textbf{Algorithm}~\ref{alg:depth_opt} focuses on minimizing the circuit depth for executing each algorithm. Additionally, we provide \textbf{Algorithm}~\ref{alg:depth_opt_sign} that thermalizes $\{(-1)^{f(b,a)}\}_{b\in\{0,1\}^k}$ in a very short depth.

\section{Thermalization of subsets}\label{sec:subset-thermalization}

In this section, we study algorithms thermalizing bits of random $t$ distinct copies of $\{0,1\}^k\times 0^{n-k}$. Before providing algorithms and related theorems, we first compute the full rank probability of a random binary matrix, which will be useful later.

\begin{lemma}\label{thm:full-rank-prob}
    The probability $p_l$ of a $l\times m$ random matrix $M$ over $\mathbb{F}_2$ with $m\gg l$ and $l=\operatorname{poly}(n)$ for some $n$ having the full rank is lower bounded by
    \begin{equation}\label{eq:full-rank-prob}
        p_l \leq \exp\left( -\frac{q^{m+1} ( q^{-l} - 1 )}{p} - \frac{sp^{\tilde{q}m+1}q^{\tilde{p}m-1}}{(1-s)^2} \right).
    \end{equation}
    with a failure probability negligible in $n$, with the non-trivial element probability $p\leq\frac{1}{4}$, $q=1-p$, $\tilde{p}=(1+\epsilon)p$, $\tilde{q}=1-\tilde{p}$, $s=(pq)^{\tilde{q}}$, and $mp=\omega(\log n)$ for some $0<\epsilon<1$.
\end{lemma}
\begin{proof}
    Let us construct a full rank matrix $M$ by sequentially sampling from the first row to the last row. Since $M$ is full rank, its rows are all linearly independent. Let us assume that we have sampled the first $r$ rows which are all linearly independent. Without loss of generality, we can reorder the columns of $M$ so that the first $r\times r$ sub-matrix of $M$ is a full rank matrix. Then, for $m\gg l$, the probability $p_{r+1}$ of sampling a row vector that is linearly independent from the previously sampled rows is given by
    \begin{equation}
        p_{r+1} = 1 - q^{m-l} - \sum_{r'=1}^l p_{r',\mathrm{ind}}.
    \end{equation}
    The second term in the right hand side is the probability of sampling a row with vanishing elements at $[r+1,m]$. $p_{r',\mathrm{ind}}$ is the probability of not sampling the previously sampled $r'$-th row. The number $N_{r'}$ of ones of the $r'$-th row is smaller than $(1+\epsilon)np$ for some small $\epsilon >0$ with a failure probability negligible in $n$ due to the Chernoff bound of
    \begin{equation}
        \operatorname{Pr}\left[N_{r'}>(1+\epsilon)mp\right] \leq \exp(-\frac{\epsilon^2 mp}{2+\epsilon}).
    \end{equation}
    Let $\tilde{p}$ and $\tilde{q}$ be $(1+\epsilon)p$ and $1-\tilde{p}$, respectively. Then, we have
    \begin{equation}
        p_{r',\mathrm{ind}} \geq p^{\tilde{q}(m-r)}q^{\tilde{p}(m-r)}\times pq^{r-1}
    \end{equation}
    with a probability higher than $1-\operatorname{negl}(n)$. This directly implies 
    \begin{equation}
        p_{r+1} \geq 1 - q^{m-r} - lp^{\tilde{q}(m-r)+1}q^{\tilde{p}m+\tilde{q}r-1}.
    \end{equation}
    as well as
    \begin{equation}
        \begin{split}
        p_l 
        &\geq \prod_{r=0}^{l-1} \left(1 - q^{m-r} - rp^{\tilde{q}(m-r)+1}q^{\tilde{p}m+\tilde{q}r-1}\right)\\
        &\approx \exp\left(-\sum_{r=0}^{l-1}\left(q^{m-r} + rp^{\tilde{q}(m-r)+1}q^{\tilde{p}m-\tilde{q}r-1}\right)\right)\\
        &= \exp\left( -\frac{1}{p} q^{m+1} ( q^{-l} - 1 ) \right)\\
        &\quad\times \exp\left( - p^{\tilde{q}m+1}q^{\tilde{p}m-1} \frac{s(1-l s^{l-1}(1-s))}{(1-s)^2} \right)
        \end{split}
    \end{equation}
    with a probability higher than $1-\operatorname{negl}(n)$. For $p\leq 1/4$ and large $l$, this becomes
    \begin{equation}
        p_l \geq \exp\left( -\frac{q^{m+1} ( q^{-l} - 1 )}{p} - \frac{sp^{\tilde{q}m+1}q^{\tilde{p}m-1}}{(1-s)^2} \right).
    \end{equation}
\end{proof}

\begin{figure}
\begin{algorithm}[H]
    \caption{Random Multi-Controlled Circuit (RMC)}
    \label{alg:RMCX}
    \begin{algorithmic}
    \Require $n > 0$, $n> x_1\geq 1$, $n\geq x_2 > x_1$, $m\geq 2$
    \Ensure $C$, $T$
    \State $l \gets 1$
    \State $C \gets$ $m$ choose $\{x_1,\cdots,x_2\}$
    \State $T \gets$ a random element of $\{0,1\}^{n-k}$
    \end{algorithmic}
\end{algorithm}
\end{figure}

\begin{figure}
\begin{algorithm}[H]
    \caption{Gate number optimized bit thermalizer}
    \label{alg:gate_num_opt}
    \begin{algorithmic}
    \Require $n > 0$, $n\geq k > 1$, $\alpha>0$, $t>0$, $m\geq 2$
    \Ensure $\ket{\psi}$
    \State $\ket{\psi} \gets \ket{+}^{\otimes k}\ket{0}^{\otimes(n-k)}$
    \State $l \gets 1$
    \While{$l \leq \alpha t$}
        \State $C, T \gets \operatorname{RMC}(n,1,k,m)$
        \State $x \gets k+1$
        \While{$x\leq n$}
            \If{$T[x-k]$}
                \State $\ket{\psi} \gets m\operatorname{-MCX}_{C,x}\ket{\psi}$
            \EndIf
            \State $x \gets x+1$
        \EndWhile
        \State $l \gets l+1$
    \EndWhile
    \State $l \gets 1$
    \While{$l \leq \alpha t$}
        \State $C, T \gets \operatorname{RMC}(n,k+1,n,m)$
        \State $x \gets 1$
        \While{$x \leq k$}
            \If{$T[x]$}
                \State $\ket{\psi} \gets m\operatorname{-MCX}_{C,x}\ket{\psi}$
            \EndIf
            \State $x \gets x+1$
        \EndWhile
        \State $l \gets l+1$
    \EndWhile
    \end{algorithmic}
\end{algorithm}
\end{figure}

\textbf{Algorithm}~\ref{alg:gate_num_opt} performs random bit-flipping of $t$ distinct copies of $\{0,1\}^k\times 0^{n-k}$ in two stages. In the first stage, it takes the first $k$ bits as control bits and the remaining $n-k$ bits as target bits for the MCX gates. For $t$ randomly chosen distinct copies of $\{0,1\}^k\times 0^{n-k}$, the first $k$ bits have random binary values for $k\ll n$, causing the MCX gates to randomly flip and thermalize the bits in positions $[k+1,n]$. In the second stage, the algorithm reverses the role of the bits, using the thermalized bits in $[k+1,n]$ as control bits and the first $k$ bits as target bits, thereby thermalizing the first $k$ bits.

The number of control bits used in \textbf{Algorithm}~\ref{alg:gate_num_opt} must be adjusted based on the number of copies $t$. This is because thermalizing the bits of $t$ copies at each site requires the use of MCX gates with $t$ distinct control conditions. In the first stage of the algorithm, since there are only $k$ bits available as control bits, the number of distinct conditions is insufficient for $t\gg k$. For instance, when using CCX gates, the maximum number of distinct conditions is $4\times \binom{k}{2}$. If $t$ exceeds this number, it becomes impossible to thermalize the bits using CCX gates, even when utilizing all possible control conditions. Therefore, the performance of the algorithm should be evaluated in two different regimes with fewer and more copies. Below, we introduce two theorems about the performance of the algorithm. \textbf{Theorem}~\ref{thm:CCX-randomization} and~\ref{thm:MCX-randomization} regard the cases with fewer copies $t=O(k)$ and more copies $t=\operatorname{poly}(n)$, respectively.

\begin{figure}[t]
    \includegraphics[width=\linewidth]{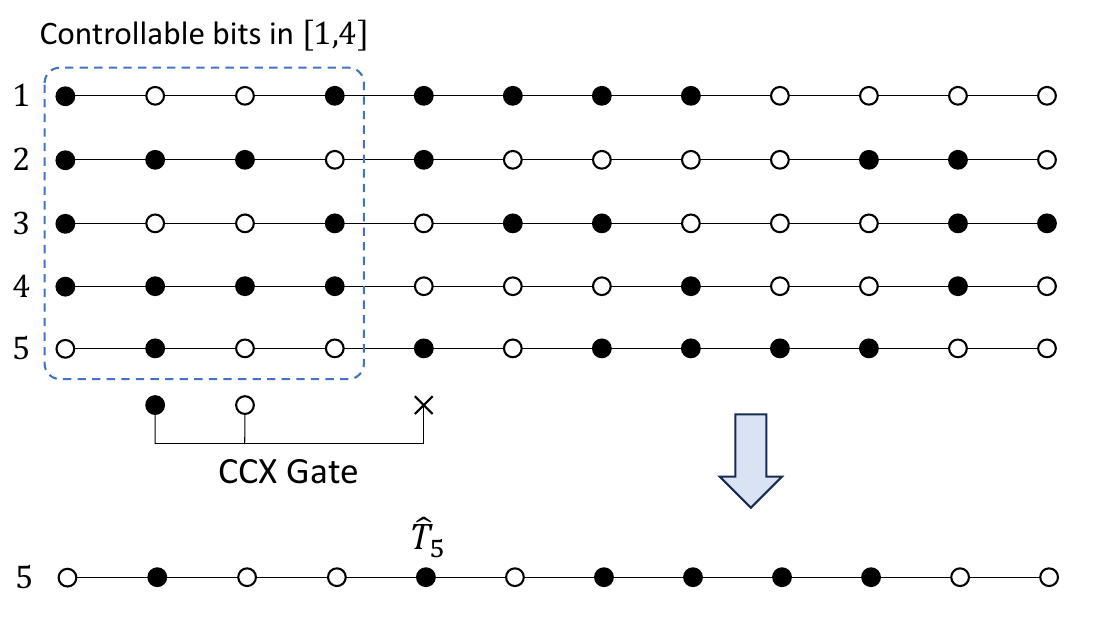}
    \centering
    \caption{ Each circle and its color represent a bit and its value. A CCX gate conditioned on second and third bits targeting to the fifth bit is applied on $\{\mathbf{x}_l\}_{l=1}^t$ with $t=5$ with the probability of $1/2$.  The random application of the CCX gate adds a $\mathbb{Z}_2$ valued random variable $\hat{T}_q$ with $q=5$ on the fifth copy. Controllable bits mean bits that can be used as control bits of CCX gates. These bits are assumed to be thermalized so that such CCX gate flip its target bit with the probability 1/4.
    }
    \label{fig:CCX-update}
\end{figure}

\begin{theorem}\label{thm:CCX-randomization}
    For $k=\omega(\log n)$ and $t\leq k/2$, \textbf{Algorithm}~\ref{alg:gate_num_opt} with $\alpha t = \omega(\log n) \leq k/2$, $\alpha\gg 1$, and $m=2$ thermalizes all bits of $t$ distinct random copies of $\{0,1\}^k\times 0^{n-k}$ with a failure probability negligible in $n$ and $\omega(n\log n)$ CCX gates.
\end{theorem}
\begin{proof}
Let us sample $t$ distinct copies $\{\mathbf{x}_l\}_{l=1}^t$ of $\{0,1\}^k\times 0^{n-k}$. Let us execute $\operatorname{RMC}(n,1,k,2)$ function $\alpha t$ times and get $\{C_l\}_{l=1}^{\alpha t}$ and $\{T_l\}_{l=1}^{\alpha t}$. Let $(x_{i_l},x_{j_l})$ be $C_l$. In addition, let us assume that $\{T_l\}_{l=1}^{\alpha t}$ be instances of random variables $\{\hat{T}_l\}_{l=1}^{\alpha t}$. Now, let us construct a $t\times\alpha t$ matrix $X$ over $\mathbb{F}_2$ whose element $(p,q)$ is one if $\mathbf{x}_p$ follows $(\mathbf{x}_p)_{i_q} = x_{i_q}$ and $(\mathbf{x}_p)_{j_q} = x_{j_q}$ and zero otherwise. In other words, if $X_{p,q}$ is non-vanishing, then bit-flips of bits in $[k+1,n]$ of the $p$-th copy are given by $\hat{T}_q$. We note different copies share the same random variable vector $\hat{T}_q$. Thus, a single CCX gate cannot flips bits of the copies independently. This is illustrated in Fig.~\ref{fig:CCX-update}. On the other hand, by using multiple CCX gates, each copy can possess an independent random variable vector which would be given by a linear combination of $\{\hat{T}_l\}_{l=1}^{\alpha t}$. This can be achieved when $X$ is full ranked. Below, we study when this happens.

Since $2^k$ is super-polynomially greater than $t$, and $t$ copies are sampled from $\{0,1\}^k\times 0^{n-k}$, their first $k$ bits are random. Thus, $X$ is a random matrix over $\mathbb{F}_2$ whose entries are independent and non-trivial with the probability of $1/4$. Due to \textbf{Lemma}~\ref{thm:full-rank-prob}, the probability $p_t$ of such a matrix having the full rank is lower bounded by Eq.~\eqref{thm:full-rank-prob} with $p=1/4$, $l=t$, $m=\alpha t$, and $0<\epsilon<1$,
\begin{equation}
    \begin{split}
    p_t &\geq \exp\left( -4\left(\frac{3}{4}\right)^{\alpha t+1}\left(\left(\frac{3}{4}\right)^{-t}-1\right) \right)\\
    &\quad\times \exp\left(-\frac{s}{(1-s)^2}\frac{3^{(1+\epsilon)\alpha/4 t-1}}{4^{\alpha t}}\right)
    \end{split}
\end{equation}
with a probability higher than $1-\operatorname{negl}(n)$. Since we have $\alpha t=\omega(\log n)$, for sufficiently large $n$, $p_t$ is given by
\begin{equation}\label{eq:CCX-suss-prob}
    p_t = 1-\operatorname{negl}(n).
\end{equation}
This shows that for sufficiently large $\alpha t$, $X$ is full-ranked, and bits in $[k+1,n]$ of the $t$ copies independently thermalize with a failure probability negligible in $n$.

Next, bits in $[1,k]$ are needed to be updated. This can be done by executing $\operatorname{RMC}(n,k+1,n,2)$ function $\alpha t$ times get $\{C'_l,T'_l\}_{l=1}^{\alpha t}$. Again, let us construct a $t\times\alpha t$ matrix $Y$ over $\mathbb{F}_2$ in the same way. Since bits in $[k+1,n]$ are fully randomized, $Y$ is a full-ranked random matrix over $\mathbb{F}_2$. Thus, CCX gates conditioned on $\{C'_l\}_{l=1}^{\alpha t}$ targeted to sites $\{1,\cdots,k\}$ fully randomize bits in $[1,k]$.
\end{proof}

\begin{theorem}\label{thm:MCX-randomization}
    For $k=\omega(\log n)$ and $t=\operatorname{poly}(n)$, \textbf{Algorithm}~\ref{alg:gate_num_opt} with $\alpha = \omega(\log n)$, and $m=\lceil \log_2 t \rceil$ thermalizes all bits of $t$ distinct random copies of $\{0,1\}^k\times 0^{n-k}$ with a failure probability negligible in $n$ and $\omega(n\log n)\cdot O(t \log t )$ CCX gates.
\end{theorem}
\begin{proof}
Let $\{\textbf{x}_l\}_{l=1}^t$ be $t$ distinct copies of $\{0,1\}^k\times 0^{n-k}$. Let us execute $\operatorname{RMC}(n,1,k,m)$ function $\alpha t$ times and get $\{C_l\}_{l=1}^{\alpha t}$ and $\{T_l\}_{l=1}^{\alpha t}$. Similar to the proof of \textbf{Theorem}~\ref{thm:CCX-randomization}, we can construct a $t\times \alpha t$ matrix $X$ over $\mathbb{F}_2$ whose element $(p,q)$ is one if $((\textbf{x}_p)_{i_{1,q}},\cdots,(\textbf{x}_p)_{i_{m,q}})=(x_{i_{1,q}},\cdots,i_{i_{m,q}})$ and zero otherwise. Each element of $X$ is one with the probability higher than $1/t$. Thus, for large enough $t$, the full rank probability of $X$ in Eq.~\eqref{eq:full-rank-prob} with $l=t$, $m=\alpha t$, and $0<\epsilon<1$ is lower bounded by
\begin{equation}
    \begin{split}
        p_t 
        &\geq \exp\left( -t \left(1-\frac{1}{t}\right)^{\alpha t+1} \left( \left(1-\frac{1}{t}\right)^{-t} - 1 \right) \right)\\
        &\quad\times\exp\left( - \frac{s(1/t)^{(t-1-\epsilon)\alpha+1}(1-1/t)^{(1+\epsilon)\alpha-1}}{(1-s)^2} \right)\\
        &\approx \exp\left( -(t-1) e^{-\alpha} \left( e - 1 \right) \right)\\
        &\quad\times\exp\left( -\frac{(t-1)^{(1+\epsilon)\alpha-1}}{t^{\alpha t}} \frac{s}{(1-s)^2} \right)\\
        &= 1 - \operatorname{negl}(n)
    \end{split}
\end{equation}
with a probability higher than $1-\operatorname{negl}(n)$. Thus, $X$ is full rank with a failure probability negligible in $n$, and all bits in $[k+1,n]$ of the $t$ copies are thermalized by applying $\alpha t$ MCX gates with $m$ conditional bits according to $\{C_l\}_{l=1}^{\alpha t}$ and $\{T_l\}_{l=1}^{\alpha t}$. The remaining bits in $[1,k]$ can be thermalized by applying MCX gates based on results of $\alpha t$ runs of $\operatorname{RMC}(n,k+1,n,m)$.
\end{proof}

The number of gates required to thermalize $t$ distinct copies of $\{0,1\}^k\times 0^{n-k}$ is $\omega(n\log n)$ for fewer copies of $t\leq k/2$. We note that this number is equivalent to that required to thermalize a single copy using the random CCX gate circuit introduced in Ref.~\cite{feng2024}. For more copies of $t=\operatorname{poly}(n)$, an additional $t$ dependence of $O(t\log t)$ is introduced as a consequence of using $\lceil\log_2 t\rceil$-MCX gates. This $t$ dependence almost saturates the theoretical lower bound of $\Omega(t)$ for thermalizing bits of the $t$ copies~\cite{feng2024}. 

On the other hand, the depth required to implement \textbf{Algorithm}~\ref{alg:gate_num_opt} is at least proportional to $n/k$, due to the limited parallelism in applying MCX gates. In the first while loop of \textbf{Algorithm}~\ref{alg:gate_num_opt}, all control bits are confined to the range in $[1,k]$, which is super-polynomially smaller than $n$ when $k=\operatorname{polylog}(n)$. Consequently, implementing \textbf{Algorithm}~\ref{alg:gate_num_opt} requires an almost extensive depth of $\Omega(n/\operatorname{polylog}(n))$. To reduce this depth, the set of bits available as control bits for MCX gates should be expanded incrementally by including sites immediately after their bits are thermalized.  Below, we present \textbf{Algorithm}~\ref{alg:depth_opt} that decomposes the thermalization process into multiple steps, each of which can be executed in parallel, and sequentially extends the control region after each step. This approach allows this algorithm to be implemented using a circuit with a sub-extensive depth of $\operatorname{polylog}(n)$ when $t$ is of order $k$. As a cost of the parallelization, this algorithm requires more CCX gates than \textbf{Algorithm}~\ref{alg:gate_num_opt} when $t=\operatorname{poly}(n)$, though the increase is limited to a multiplicative $\log t$ factor. In addition to the algorithm, we also provide \textbf{Theorem}~\ref{thm:CCX-low-depth-randomization} and~\ref{thm:MCX-low-depth-randomization} which consider the cases with $t=O(k)$ and $t=\operatorname{poly}(n)$ copies, respectively.

\begin{figure}
\begin{algorithm}[H]
    \caption{Parallel RMC (pRMC)}
    \label{alg:pRCCX}
    \begin{algorithmic}
    \Require $n > 0$, $n> x_1\geq 1$, $n\geq x_2 > x_1$, $m\geq 2$, $p > 0$
    \Ensure $C$, $T$
    \State $l \gets 1$
    \State $k \gets x_2-x_1+1$
    \State $Cp \gets$ $mp$ choose $\{x_1,\cdots,x_2\}$
    \State $C \gets$ random $p$ equal partitions of $Cp$
    \State $T \gets$ random $p$ elements of $\{0,1\}^{n-k}$
    \end{algorithmic}
\end{algorithm}
\end{figure}

\begin{figure}
\begin{algorithm}[H]
    \caption{Depth optimized bit thermalizer}
    \label{alg:depth_opt}
    \begin{algorithmic}
    \Require $n > 0$, $n\geq k > 1$, $\alpha>0$, $t>0$, $m\geq 2$
    \Ensure $\ket{\psi}$
    \State $\ket{\psi} \gets \ket{+}^{\otimes k}\ket{0}^{\otimes(n-k)}$
    \State $s \gets k$
    \While{$s < n$}
        \State $p \gets \lfloor s/m \rfloor$
        \State $l \gets 1$
        \While{$l \leq \alpha t$}
            \State $C, T \gets \operatorname{pRMC}(n,1,s,m,p)$
            \State $x \gets 1$
            \While{$x\leq p$} \Comment{This can be done in parallel.}                
                \If{$T[x]$}
                    \State $\ket{\psi} \gets m\operatorname{-MCX}_{C[x],s+x}\ket{\psi}$
                \EndIf
                \State $x \gets x+1$
            \EndWhile
            \State $l \gets l+1$
        \EndWhile
        \State $s \gets s+p$
    \EndWhile
    \State $p \gets \lfloor (n-k)/m \rfloor$
    \State $l \gets 1$
    \While{$l \leq \alpha t$}
        \State $C, T \gets \operatorname{pRMC}(n,k+1,n,m,p)$
        \State $x \gets 1$
        \While{$x\leq p$} \Comment{This can be done in parallel.}
            \If{$T[x]$}
                \State $\ket{\psi} \gets m\operatorname{-MCX}_{C[x],s+x}\ket{\psi}$
            \EndIf
            \State $y \gets y+1$
        \State $x \gets x+1$
        \EndWhile
        \State $l \gets l+1$
    \EndWhile
    \end{algorithmic}
\end{algorithm}
\end{figure}

\begin{theorem}\label{thm:CCX-low-depth-randomization}
    For $k=\omega(\log n)$ and $t\leq k/2$, \textbf{Algorithm}~\ref{alg:depth_opt} with $\alpha t = \omega(\log n) \leq k/2$, and $m=2$ thermalizes $t$ distinct random copies of $\{0,1\}^k\times 0^{n-k}$ in $\omega(\log n)$ depth with a failure probability negligible in $n$.
\end{theorem}
\begin{proof}
    Due to \textbf{Theorem}~\ref{thm:CCX-randomization}, the first iteration of applying CCX gates according to $\operatorname{pRMC}(n,1,k,m,\lfloor k/2\rfloor)$ theramlizes bits in $[s_1+1,s_2]$ of the $t$ copies with $s_1=k$ and $s_2=(3/2)k$. Let us assume that we have completed the $(i-1)$-th iteration. This implies that bits of $t$ copies in $[s_1+1,s_i]$ are thermalized. Thus, the matrix $X$ defined in the proof of \textbf{Theorem}~\ref{thm:CCX-randomization} constructed from randomly partitioned control bits in $[1,s_i]$ is again a random matrix over $\mathbb{F}_2$. Thus, the $i$-th iteration thermalizes bits in $[s_i+1,s_{i+1}]$. By the mathematical induction, \textbf{Algorithm}~\ref{alg:depth_opt} flips all bits in $[k+1,n]$ of $t$ distinct random copies of $\{0,1\}^k\times 0^{n-k}$.

    At the $i$-th iteration, \textbf{Algorithm}~\ref{alg:depth_opt} applies $\alpha t$ CCX gates on each of bits in $[s_i+1,s_{i+1}]$ conditioned on bits in $[1,s_i]$. By finding $\alpha t$ random coverings of $[1,s_i]$ by non-overlapping pairs of bits, $s_i/2$ bits can be simultaneously targeted and randomly flipped. We note that this way to choose pairs of conditional bits does not affect the proof of \textbf{Theorem}.~\ref{thm:CCX-low-depth-randomization}. Since we have $s_i/2=s_{i+1}-s_i$, CCX gates conditioned on bit pairs of each covering update all bits in $[s_i+1,s_{i+1}]$ simultaneously. Thus, the depth of a circuit implementing the $i$-th iteration is given by $\alpha t$. At each iteration, the number of bits that can be used as controlled bits is increasing by $3/2$ times. Thus, the sum of the depth of each circuit implementing each iteration is $O(\alpha t \log n)$. Since we have, $\alpha t=\omega(\log n)$, the depth is $\omega(\log n)$.
\end{proof}

\begin{theorem}\label{thm:MCX-low-depth-randomization}
    For $k=\omega(\log n)$ and $t=\operatorname{poly}(n)$, \textbf{Algorithm}~\ref{alg:depth_opt} with $\alpha = \omega(\log n)$, and $m=\lceil \log_2 t \rceil$ thermalizes $t$ distinct random copies of $\{0,1\}^k\times 0^{n-k}$ in $\omega(\log n)O\left(t [\log (t)]^2\right)$ depth with a failure probability negligible in $n$.
\end{theorem}
\begin{proof}
    A similar proof of \textbf{Theorem}~\ref{thm:CCX-low-depth-randomization} can be applied here. Due to \textbf{Theorem}~\ref{thm:MCX-randomization}, the first iteration generates random bits of $t$ copies in $[s_1+1,s_2]$ with $s_1=k$ and $s_2=(1+\lfloor 1/m\rfloor )k$. Let us assume that we have completed the $(i-1)$-th iteration. Then, the matrix $X$ defined in the proof of \textbf{Theorem}~\ref{thm:MCX-randomization} constructed from randomly sampled sets of $\log_2 t$ control bits in $[1,s_i]$ is a random matrix over $\mathbb{F}_2$, and the $i$-th iteration flips all bits in $[s_i+1,s_{i+1}]$ with the probability of $1/2$. By the mathematical induction, \textbf{Algorithm}~\ref{alg:depth_opt} flips all bits in $[k+1,n]$ of $t$ distinct random copies of $\{0,1\}^k\times 0^{n-k}$.

    At the $i$-th iteration, \textbf{Algorithm}~\ref{alg:depth_opt} applies $m$-MCX gates $\alpha t$ times targeting each of bits in $[s_i+1,s_{i+1}]$ conditioned on bits in $[1,s_i]$. By finding $\lfloor s_i/m\rfloor$ random equal partitions of $[1,s_i]$ with the partition size $m$, all bits in $[s_i+1,s_{i+1}]$ can be simultaneously targeted and flipped randomly. In other words, the depth of a circuit implementing the $i$-th iteration is given by $\alpha t$. 
    
    At each iteration, the number of bits that can be used as control-bits is increasing by $(1+\lfloor 1/m \rfloor)$ times. Then, the total number of iterations required to get $n$ controllable bits is asymptotically given by $\log n \log_2 t$ since we have
    \begin{equation}
        n \approx \left(1+\frac{1}{\log_2 t}\right)^{\log n \log_2 t}
    \end{equation}
    for sufficiently large $n$. Since each MCX gate with $m$ control bits can be implemented by $O(m)$ CCX gates, the total depth of a circuit implementing \textbf{Algorithm}~\ref{alg:depth_opt} is given by $O\left(\alpha t \log (n) [\log (t)]^2\right)=\omega(\log n)O(t[\log (t)]^2)$.
\end{proof}

\section{Thermalization of signs}\label{sec:sign-thermalization}

We have introduced two algorithms for generating random subsets. However, to create random subset phase states, we need to additionally implement random signs of the state $\ket{\psi}$ in Eq.\eqref{eq:random-subset-phase-state}. An approximate implementation of these random signs, up to the $t$-th moment, can be achieved by applying signed MCZ gates randomly after thermalizing $t$ distinct copies of $\{0,1\}^n$ using, for example, \textbf{Algorithm}~\ref{alg:gate_num_opt} or~\ref{alg:depth_opt}. A signed MCZ gate flips the sign of a computational state if its condition is satisfied, and its target bit has $0$ or $1$ depending on the sign of the gate. Effectively, it can be thought of as a MCX gate targeting an auxiliary bit that encodes signs of the state. Thus, similar to the bit-flip algorithms, random signed MCZ gates can thermalize signs of the $t$ copies with a circuit depth shorter than that of \textbf{Algorithm}~\ref{alg:depth_opt}. Consequently, the depth required to generate $\ket{\psi}$ is dominated by the depth needed to thermalize the subset distribution. 

As a side note, a subset phase state ensemble without any random phase factors is already a pseudorandom ensemble when the subset size is neither too small nor too large~\cite{giurgicatiron2023}. Therefore, the ability to randomize the subsets alone is sufficient to generate a pseudorandom ensemble.

Below, we present $\textbf{Algorithm}$~\ref{alg:depth_opt_sign}, which implements a random sign operator $[F]_{b,b'}=\delta_{b,b'}(-1)^{f(b)}$, where $f$ is a random function $f:\{0,1\}^n\rightarrow \{0,1\}$, in a depth-optimized manner. This random sign operator is crucial for ensuring that the random subset states exhibit pseudorandomness when $k=\omega(\log n)$. More precisely, the exactness of the $t$-th moments of random phases is used to show that the random subset phase state ensemble forms a state $t$-design with a negligible error in \textbf{Lemma 2.4} of Ref.~\cite{aaronson2023quantum}.

\begin{figure}
\begin{algorithm}[H]
    \caption{Depth optimized sign thermalizer}
    \label{alg:depth_opt_sign}
    \begin{algorithmic}
    \Require $n > 0$, $p>0$, $\alpha>0$, $t>0$, $m\geq 2$, $\ket{\psi}$
    \Ensure $\ket{\psi}$
    \State $m' \gets m-1$
    \State $l \gets 1$
    \While{$l \leq \lceil \alpha t / p \rceil$}
        \State $C, T \gets \operatorname{pRMC}(n,1,mp,m,p)$
        \State $x \gets 1$
        \While{$x\leq p$}\Comment{This can be done in parallel.}
            \If{$T[x]$}
                \State $\ket{\psi} \gets m'\operatorname{-MCZ}_{(C[x][1],\cdots,C[x][m']),mx}\ket{\psi}$
                \State $\ket{\psi} \gets (-1)^{C[x][m]} \ket{\psi}$
            \EndIf
            \State $x \gets x+1$
        \EndWhile
        \State $l \gets l+1$
    \EndWhile
    \end{algorithmic}
\end{algorithm}
\end{figure}

\begin{theorem}\label{thm:low-depth-sign-randomization}
    For $t\leq n$, \textbf{Algorithm}~\ref{alg:depth_opt} with $p=n$, $\alpha t = \omega(\log n) \leq n$, and $m=1$ thermalizes signs of random distinct $t$ copies of $\{0,1\}^n$ in $O(1)$ depth with a failure probability negligible in $n$.
\end{theorem}
\begin{proof}
    Let us sample $t$ distinct copies $\{\textbf{x}_l\}_{l=1}^t$ of $\{0,1\}^n$. When $m=1$, a MCZ gate is given by the Pauli $Z$ operator at the control bit. Thus, up to $n$ MCZ gates can be applied in parallel. Next, let us construct the matrix $X$ defined in the proof of \textbf{Theorem}~\ref{thm:CCX-randomization}. Different from the proof, here, each element of $X$ is related to a random variable (vector) flipping the signs, not the bits, of the $t$ copies. In addition, since the copies are randomly sampled from $\{0,1\}^n$, the probability of each element of $X$ being non-trivial is $1/2$, which differs from the probability set by the proof. However, it does not alter the scaling of the full rank probability of $X$. Therefore, $\alpha t = \omega(\log n)$ applications of MCZ gates thermalize the signs of the copies.
\end{proof}

This theorem shows the efficiency of \textbf{Algorithm}~\ref{alg:depth_opt_sign} which thermalizes the signs of a random subset phase state by a unit depth circuit. However, when $t$ is larger than $n$, it becomes impossible to use this theorem as there does not exist a positive integer $\alpha$ such that $\alpha t \leq n$. For a generic $t=\operatorname{poly}(n)$, we again need to use multi-controlled gates with $m=\lceil\log_2 n\rceil$.

\begin{theorem}\label{thm:low-depth-sign-randomization}
    For $t=\operatorname{poly}(n)$, \textbf{Algorithm}~\ref{alg:depth_opt} with $p=\lfloor n/\lceil\log_2 n\rceil\rfloor$, $\alpha =\omega(\log n)$, and $m=\lceil\log_2 n\rceil$ thermalizes signs of random distinct $t$ copies of $\{0,1\}^n$ in $\omega(\log n)\cdot O(t/n)$ depth with a failure probability negligible in $n$.
\end{theorem}
\begin{proof}
    First, $p$ is the maximum number of MCZ gates with $m$ conditional bits. Following the proof of \textbf{Theorem}~\ref{thm:MCX-randomization}, the matrix $X$ has full rank with a negligible failure probability when $\alpha=\omega(\log n)$ MCZ gates are applied. Here, similar to the proof of \textbf{Theorem}~\ref{thm:low-depth-sign-randomization}, elements of $X$ are related to sign-flips of a random subset phase state $\ket{\psi}$. Since the depth of \textbf{Algorithm}~\ref{alg:depth_opt} is given by $\lceil\alpha t/p\rceil$, it thermalizes signs of $\ket{\psi}$ in $\omega(\log n)\cdot O(t/n)$ depth.
\end{proof}

\section{From erroneous sampling to pseudorandomness}

In Sec.~\ref{sec:subset-thermalization} and~\ref{sec:sign-thermalization}, we present algorithms for thermalizing subsets and the signs of random subset phase states. In addition, we demonstrate that they can be implemented using circuits with a small number of gates or very short depth, with negligible failure probabilities. However, this does not directly imply that the resulting random subset phase state ensemble is pseudorandom. The following theorem shows that the former, the erroneous sampling, indeed implies the latter, pseudorandomness of the resulting ensemble.

\begin{theorem}
    Let $t$ be the number of copies given by $t=\operatorname{poly}(n)$. Let $\sigma$ be the $t$ copies of the random subset phase state ensemble, and let $\sigma_\mathrm{E}$ be an ensemble of $t$ copies of subset phase states whose $t$-th moments of subsets and signs are fail to be thermalized. Let $p_\mathrm{fail}$ be the probability of getting a state in $\sigma_\mathrm{E}$, and let $\rho$ be the density matrix description of this sampling process, which is given by
    \begin{equation}
        \rho = (1-p_\mathrm{fail})\sigma + p_\mathrm{fail} \sigma_\mathrm{E}.
    \end{equation}
    Then, if $p_\mathrm{fail}$ is negligible, then the trace distance between $\rho$ and the ensemble average of Haar random states $\rho_\mathrm{Haar}=\int d\psi_\mathrm{Haar}\ketbra{\psi}{\psi}^{\otimes t}$ is also negligible.
\end{theorem}
\begin{proof}
    Let $\operatorname{TD}(\rho,\sigma)$ be the trace distance between $\rho$ and $\sigma$. Then, due to the convexity of the trace distance, we have
    \begin{equation}
        \begin{split}
            \operatorname{TD}(\rho,\rho_\mathrm{Haar}) 
            &\leq  p_\mathrm{fail}\operatorname{TD}(\sigma_\mathrm{E}, \rho_\mathrm{Haar}) \\
            &\quad + (1-p_\mathrm{fail})\operatorname{TD}(\sigma, \rho_\mathrm{Haar}).
        \end{split}
    \end{equation}
    Due to \textbf{Theorem 2.1} of Ref.~\cite{aaronson2023quantum}, $\operatorname{TD}(\sigma, \rho_\mathrm{Haar})$ is negligible in $n$. In addition, $\operatorname{TD}(\sigma_\mathrm{E}, \rho_\mathrm{Haar})$ is upper bounded by unity. Then, it follows that
    \begin{equation}
        \begin{split}
            \operatorname{TD}(\rho,\rho_\mathrm{Haar}) 
            &\leq p_\mathrm{fail} + \operatorname{negl}(n) \\
            &= \operatorname{negl}(n).
        \end{split}
    \end{equation}
\end{proof}

\section{Comparison with previous works}\label{sec:previous-work-comparision}
Here, we compare performance of our algorithm for generating pseudorandom states with that of previous works. To the best our knowledge, our work provides the fastest algorithm, the combination of \textbf{Algorithm}~\ref{alg:depth_opt} and~\ref{alg:depth_opt_sign}, for generating pseudorandom states, which only requires $\omega(\log n)\cdot O(t[\log t]^2)$ depth. Below, we briefly review the previous works.

Pseudorandom states have been produced in various ways. One approach is to use local random circuits. In $D$-dimension, local random circuits can generate pseudorandom states in $\operatorname{poly}(t)\cdot n^{1/D}$ depth, where $\operatorname{poly}(t)$ is much larger than $O(t)$~\cite{Brand_o_2016,Haferkamp2022randomquantum,Harrow_2023}. Another way to generate pseudorandom states is to use projected ensembles, which require an evolving time of $[f(t)]^{1/\alpha} \omega(\operatorname{poly}(n))$ for some increasing function $f(t)$ and constant $\alpha$~\cite{Cotler2023}. The copy-dependence of the depth can be reduced to $O(t^2\cdot n\log n)$ by using random Pauli rotations~\cite{haah2024}. This was conjectured to be further reduced to $\omega(n\log n)\cdot O(t)$ using random automaton circuits~\cite{feng2024}. Very recently, an almost linear $t$ dependence of the depth has been achieved by gluing small random unitary designs, producing pseudorandom states in $\omega(\log n)\cdot O(t [\log t]^7)$ depth without using ancillary qubits~\cite{schuster2024}. Finally, pseudorandom states with odd prime qudits can be generated by $O(\log n)$ depth circuits by using inflationary quantum gates~\cite{Chamon2024}. However, there is no known such gate in qubit systems, so the circuit depth for generating pseudorandom states is lower bounded by $\omega(\log n)$, which is saturated by our results.

\section{Conclusion}
In this work, we provide algorithms that approximately generate random subset phase states by circuits with small numbers of gates or short depth, greatly improving upon previous methods. This directly implies that our algorithms can be used to simulate Haar random states, which has been widely used in various applications, including classical shadow tomography, benchmarking quantum circuits, demonstrating quantum supremacy, without compromising their quality, assuming only a polynomial number of measurements is allowed. Furthermore, our algorithms can be used to implement pseudo-chaotic dynamics in a very short depth~\cite{lee2024pseudo,gu2024chaos}. We believe this opens up the possibility of performing these tasks on near-term quantum devices.

\begin{acknowledgments}
We thank Byungmin Kang for helpful discussions. W.L. and G.Y.C. are supported by Samsung Science and Technology Foundation under Project Number SSTF-BA2002-05 and SSTF-BA2401-03, the NRF of Korea (Grants No.~RS-2023-00208291, No.~2023M3K5A1094810, No.~2023M3K5A1094813, No.~RS-2024-00410027, No.~RS-2024-00444725) funded by the Korean Government (MSIT), the Air Force Office of Scientific Research under Award No.~FA2386-22-1-4061, and Institute of Basic Science under project code IBS-R014-D1. H.K. is supported by the KIAS Individual Grant No. CG085301 at Korea Institute for Advanced Study and National Research Foundation of Korea (Grants No. 2023M3K5A109480511 and No. 2023M3K5A1094813).
\end{acknowledgments}

\bibliography{Ref}

\end{document}